\documentclass[11pt]{article}
\usepackage[T1]{fontenc}
\usepackage{lmodern}
\usepackage[utf8]{inputenc}
\usepackage{mdframed}
 
\usepackage[paper=letterpaper, margin=1in]{geometry}

\usepackage{authblk}
 
\usepackage{amsmath}
\usepackage{amssymb}
\usepackage{amsthm}
\usepackage{graphicx}
\usepackage{xspace}
\usepackage{caption}
\usepackage{subcaption}
\usepackage{thm-restate}

\makeatletter
\newtheorem{repeattheorem@}{Theorem}
\newenvironment{repeattheorem}[1]{%
    \def\therepeattheorem@{\ref{#1}}
    \repeattheorem@
}
{\endrepeattheorem@}
\makeatother

\usepackage{xparse}
\usepackage{todonotes}

\newtheorem*{theorem*}{Theorem}
\newtheorem{theorem}{Theorem}
\newtheorem{lemma}[theorem]{Lemma}

\theoremstyle{definition}
\newtheorem{definition}[theorem]{Definition}
\theoremstyle{plain}

\newcommand{\distance}[2]{\ensuremath{\mathrm{Dist}(#1, #2)}}

\usepackage{xparse} 

\makeatletter

\newcommand{\anext}[1]{a_{#1}^{\textrm{next}}}

\makeatletter
\def\whp{w.h.p\@ifnextchar.{}{\@ifnextchar,{.}{.\ }}}
\makeatother

\date{}

\title{A Tight Bound for Asynchronous Collaborative Grid Exploration}

\author[1]{Sebastian Brandt}
\author[1, 2]{Jara Uitto}
\author[1]{Roger Wattenhofer}
\affil[1]{ETH Zurich, Switzerland\\
  \texttt{brandts@ethz.ch}\\
  \texttt{jara.uitto@inf.ethz.ch}\\
  \texttt{wattenhofer@ethz.ch}}
\affil[2]{University of Freiburg, Germany}

\begin{document}

\maketitle
\begin{abstract}
	Recently, there has been a growing interest in grid exploration by agents with limited capabilities.
	We show that the grid cannot be explored by three semi-synchronous finite automata, answering an open question by Emek et al.~\cite{Emek15} in the negative.
	
In the setting we consider, time is divided into discrete steps, where in each step, an adversarially selected subset of the agents executes one look-compute-move cycle.
	The agents operate according to a shared finite automaton, where every agent is allowed to have a distinct initial state.
	The only means of communication is to sense the states of the agents sharing the same grid cell.
	The agents are equipped with a global compass and whenever an agent moves, the destination cell of the movement is chosen by the agent's automaton from the set of neighboring grid cells.
	In contrast to the four agent protocol by Emek et al., we show that three agents do not suffice for grid exploration.
\end{abstract}

\newpage

\section{Introduction}
Consider the problem of exploring an infinite grid with a set of mobile robots, ants, or \emph{agents}.
In practical applications, it is often desirable to make use of inexpensive and simple devices and therefore, a finite automaton is an attractive choice for modeling these agents.
Furthermore, neither reliable communication nor synchronous time is always available and thus, distributed and non-synchronous solutions are needed.
Also exploration models inspired by biology require these features; for example models for ant foraging assume limited capabilities and distributed searching.
In both settings mentioned above, it is often reasonable to assume simple means of communication of nearby agents.

\paragraph*{Semi-Synchrony.}
Recently, there has been a growing interest in studying constant memory agents performing exploration on an infinite grid.
An infinite grid is a natural discrete version of a plane which disallows the bounded memory agents to make any use of the boundaries of the grid.
Emek et al.~\cite{Emek2014} introduced a model where the agents are able to communicate by sensing each other's states and showed a tight upper bound for the time needed for $k$ agents to find a treasure\footnote{In the deterministic case, exploring the grid and finding a treasure are equivalent. 
In the randomized case, considering a treasure is more convenient as the exploration is equivalent to hitting every cell in expected finite time.} at distance $D$.
As the first step into the model, let us introduce the way that the semi-synchrony is defined.
The time is divided into discrete time steps, and in each time step, an adversarially chosen subset of the agents performs a look-compute-move cycle in parallel.
In each cycle, the chosen agents first sense the states of all the other agents in the same cell and then, determined by their transition function, either stay still or move to an adjacent grid cell.
We point out that in every step, every agent performs the ``look'' action before any agent executes their ``compute'' step, i.e., agents sharing a cell and activated in the same time step see each other's states before any of them executes a state transition.
This definition allows an arbitrary discrepancy in the number of steps the agents perform but ensures that, whenever two agents meet, at least one of them will be able to sense the presence of the other agent.

All input parameters, such as $D$ and $k$ are unknown to the agents and they are all initially located in the origin of the grid.
Motivated by the fact that ants are able to perform very precise path integration, it is assumed that the agents are endowed with a global compass.

\paragraph*{Previous Results.}
Following up on the above model, Emek et al.~\cite{Emek15} studied the minimum number of agents needed to explore the infinite grid, where exploring refers to reaching any fixed cell within (expected) finite time.
They showed that three randomized and four deterministic semi-synchronous agents are enough for the exploration task. We want to point out that the \emph{asynchronous} environment in their paper is referred to as semi-synchronous in older literature~\cite{Sugihara1996, Suzuki1999}.
The paper left two open questions:
\begin{center}
	\begin{minipage}{0.9\textwidth}
		\emph{Can two agents controlled by a randomized FA solve the synchronous or asynchronous version of the ANTS\,\footnotemark\ problem?} \\ 
		\emph{Is there an effective FA-protocol for async-ANTS for three agents when no random bits are available?}
	\vspace{-8pt}
	\end{minipage}
\end{center}
\footnotetext{The ANTS problem in their context is the same as our grid exploration problem.}\bigskip
Very recently, Cohen et al. solved the first question by showing that two randomized agents do not suffice~\cite{Cohen2017}.
The main result of this paper is a negative answer to the second question:
\begin{theorem}
	Three semi-synchronous agents controlled by a finite automaton are not sufficient to explore the infinite grid.
	\label{thm:threelower}
\end{theorem}

Our result is obtained by solving two technical challenges.
First, we carefully design an adversarial schedule for the agents that, under the assumption that the agents actually explore the entire grid, forces them to obey a movement pattern with the following property: There is a fixed width $w$ and fixed slope $s$ auch that at any point in time, all agents are contained in a band of width $w$ and slope $s$.
Second, we formally show that the agents cannot encode a super-constant amount of information in their relative positions. 
In other words, while the relative distance can be unbounded and represent an unbounded amount of information, we can bound the amount of information the agents can infer from their relative positions.
Due to space constraints, most of our proofs are deferred to the full version of the paper~\cite{BrandtUW17}.

\section{Related Work}
Graph exploration is a widely studied problem in the computer science literature.
In the typical setting one or more agents are placed on some node of a graph and the goal is to visit every node and/or edge of the graph by moving along the edges.
There is a wide selection of variants of graph exploration and one of the standard ways to classify these variants is to divide them into \emph{directed} and \emph{undirected} variants~\cite{Deng1999, Albers2000}.
In the directed model, the edges of the graph only allow traversing into one direction, whereas in the undirected model, traversing both ways is allowed.
Our work assumes the undirected graph exploration model.

Other typical parameters of the problem are the conditions of a successful exploration and symmetry breaking mechanisms.
Some related works demand that the agents are required to halt after a successful exploration~\cite{Diks2004} or that the agents must return to their starting point after the exploration~\cite{Averbakh1996}.
From the perspective of symmetry breaking, one characterization is to break the problem into the case of equipping nodes with unique identifiers~\cite{Panaite1998, Duncan2006} and into the case where nodes are anonymous~\cite{Budach1978, Rollik1979, Bender1994}.
Since the memory of our agents is restricted to a constant amount of bits with respect to the size of the graph, the unique identifiers are not helpful.


The agents typically operate in \emph{look-compute-move} cycles, where they first gather the local information, then perform local computations, and finally, decide to which node they move.
This execution model can be divided into \emph{synchronous}~\cite{Suzuki1999}, \emph{semi-synchronous}~\cite{Sugihara1996, Suzuki1999} and \emph{asynchronous} variants~\cite{Suzuki96, Flocchini2000}, referred to as $\mathcal{FSYNC}$, $\mathcal{SSYNC}$, and $\mathcal{ASYNC}$.
In the $\mathcal{FSYNC}$ model, all agents execute their cycles simultaneously in discrete rounds.
In the $\mathcal{SSYNC}$ model only a subset (not necessarily proper) of the agents is \emph{activated} in every round and in the $\mathcal{ASYNC}$ model, the cycles are not assumed to be atomic.
To avoid confusion, we refer to the non-synchronous rounds as time steps.
In this paper, we consider the semi-synchronous model.
Note that since the $\mathcal{ASYNC}$ model is weaker than the $\mathcal{SSYNC}$ model, we directly obtain our lower bound result for the $\mathcal{ASYNC}$ model as well.

The standard efficiency measure of a graph exploration algorithm executed in the $\mathcal{FSYNC}$ model is the number of synchronous rounds it takes until the graph is explored~\cite{Panaite1998}.
In the non-synchronous models, this measure is typically generalized to the maximum delay between activation times of any agent~\cite{Chrobak2015}.
A widely-studied classic is the \emph{cow-path} problem, where the goal of the cow is to find food or a treasure on a line as fast as possible.
There is an algorithm with a constant competitive ratio for the case of a line and in the case of a grid, a simple spiral search is optimal and the problem has been generalized to the case of many cows~\cite{Baeza-Yates1993, Lopez-Ortiz2001}.
Some more recent work studied the time complexity of $n$ distributed agents searching for a treasure in distance $D$ on a grid and a $\Theta(D/n^2 + D)$ bound was shown in the case of Turing machines without communication and in the case of communicating finite automata~\cite{Feinerman2012, Emek2014}.

Our work does not focus on the time complexity of the problem, but rather on the computability, i.e, what is the minimum number of agents that are required to find the treasure.
The canonical algorithm in the case of little memory is the random walk, where the classic result states that a random walk explores an $n$-node graph in polynomial time~\cite{Aleliunas1979}.
In the case of infinite grids, it was shown in a recent paper that, even with a globally consistent orientation, two randomized agents cannot locate the treasure in finite expected time~\cite{Cohen2017}.
By combining this result with previous work~\cite{Blum1977, Emek15}, it follows that this lower bound is tight.
In the deterministic case, our lower bound of three deterministic semi-synchronous agents closes the remaining gap in the results of~\cite{Emek15}.

Another typical measure for efficiency is the number of bits of memory needed per agent~\cite{Fraigniaud2004, Diks2004}.
For example, it was shown by Fraigniaud et\ al., that $\Theta(D \log \Delta)$ bits are needed for a single agent to locate the treasure, where $D$ and $\Delta$ denote the diameter and the maximum degree of the graph, respectively.
The memory of our agents is bounded by a universal constant, independent of any graph parameters.

Work that falls close to our work is the study of graph exploration in \emph{labyrinths}, i.e., graphs that can be seen as 2-dimensional grids, where some subset of the nodes cannot be entered by the agents.
The classic results state that all co-finite (finite amount of cells not blocked) labyrinths can be explored by two finite automata and an automaton with two pebbles~\cite{Blum1978}, and that finite labyrinths (finite amount of cells are blocked) can be explored using one agent with four pebbles~\cite{Blum1977}, where a pebble is a movable marker.
Furthermore, it is known since long that there are finite and co-finite labyrinths where one pebble is not enough~\cite{Hoffmann1981} and that no finite set of finite automata can explore all planar graphs~\cite{Rollik1979}.
More recently, it was shown that $\Theta(\log \log n)$ pebbles for an agent with $\Theta(\log \log n)$ memory is the right answer for general graphs~\cite{Disser2016}.
Notice that since we do not assume synchronous communication between agents and a pebble can always be simulated by a finite automaton, our result also yields the same bound for the pebble model.

\section{Preliminaries}

\subsection*{The Model} 
The model we use is the same as in~\cite{Emek15}.
We consider a group of $n$ agents whose task is to explore every cell of the infinite $2$-dimensional grid where a cell is considered as explored when it has been visited by at least one of the agents.
We identify each cell of the grid with a pair of integers, i.e., the grid can be considered as $\mathbb Z^2$, with two cells being \emph{neighbors} if and only if they differ in one coordinate by exactly $0$ and in the other coordinate by exactly $1$.

In the beginning, all agents are placed in the same cell, called the \emph{origin}.
W.l.o.g., we will assume that the origin has the coordinates $(0, 0)$.
For the agents, all cells, including the origin, are indistinguishable; in particular, they do not have access to the coordinates of the cells.

Each agent is endowed with a \emph{compass}, i.e., each agent is able to distinguish between the four (globally consistent) cardinal directions in any cell and all agents have the same notion of those directions.
The behavior of each agent is governed by a deterministic finite automaton.
While we allow the agents to use different finite automata, we will assume that the agents use the same finite automaton but have different initial states.
Since in all cases we consider, $n$ is a constant, the two formulations are equivalent.

The only way in which communication takes place is the following:
Each agent senses for any state $q$ of the finite automaton whether there is at least one other agent in the same cell in state $q$.
In each step of the execution, an agent moves to an adjacent cell or stays in the current cell, solely based on its current state in the finite automaton and the subset of states $q$ for which another agent in state $q$ is present in the current cell. 

Given the above, we are set to describe our finite automaton more formally.
Let $Q$ denote the set of states, with each agent having its own initial state in $Q$.
The set of input symbols is $2^Q$, the set of all subsets of $Q$, reflecting the fact that for each state from $Q$ an agent in this state might be present or not in the considered cell.
The transition function $\delta: Q \times 2^Q \rightarrow Q \times \{ 0, 1, 2, 3, 4 \}$ provides an agent in state $q \in Q$ (sensing a subset $Q' \subseteq Q$ of states present in the same cell) with a new state $q' \in Q$ and a movement, where $1, 2, 3, 4$ stand for the four cardinal directions while $0$ indicates that the agent stays in the current cell. 

The $\mathcal{SSYNC}$~\cite{Sugihara1996, Suzuki1999} environment in which the agents perform their exploration is \emph{semi-synchronous}.
More specifically, we assume that the order of the steps of the agents is determined by an \emph{adversarial scheduler} that knows the finite automaton governing the agents' behavior.
Each step of an agent is a complete \emph{look-compute-move} cycle, where first an agent senses for which states agents are present in the current cell, then it applies the transition function with the sensed states and its own current state as input, and finally it moves as indicated by the result.
Cycles of different agents may occur at the same time, in which case each of the agents completes the sensing before any of the agents starts to move.
Cycles that do not occur at the same time have no overlap, i.e., the movement performed in an earlier cycle is completed before the sensing in a later cycle starts.
Hence, we may consider the order of the individual components of the execution as given by a mapping of the agents' steps to points in time.

We call such a mapping a \emph{schedule}.
Since the look-compute-move cycles of the agents are atomic in nature, we can assume w.l.o.g.\ that the static configurations of the agents on the grid (including the information about the states they are currently in) occur at integer points in time $t = 0, 1, \dots$, and that the steps of the agents determining the transition from one configuration to a new one take place between these points in time.
If an agent's action is scheduled between time $t$ and $t+1$, we say, for the sake of simplicity, that the action takes place at time $t$.
In order to prevent the adversary from delaying a single agent indefinitely, we adopt the common requirement that each agent is scheduled infinitely often.
For our lower bound we will only use adversarial schedules where no two agents are scheduled at the same time.

\subsection*{Definitions and Notation}
For the notion of distance between two cells we will use the Manhattan distance.
Let $c=(x,y)$, $c'=(x',y')$ be two cells of the infinite grid.
Then, the \emph{distance} between $c$ and $c'$ is defined as $\distance{c}{c'} = |x - x'| + |y - y'|$.
Moreover, we call the first coordinate of a cell the \emph{$x$-coordinate} and the second coordinate the \emph{$y$-coordinate}.
We denote the cell an agent $a$ occupies at time $t$ by $c_t(a) = \left (x_t(a), y_t(a) \right )$.
Similarly, we denote the state of the finite automaton in which agent $a$ is at time $t$ by $q_t(a)$.
If $a = a_i$ for some $1 \leq i \leq 3$, then we also write $c_t^i, x_t^i, y_t^i, q_t^i$ instead of $c_t(a_i), x_t(a_i), y_t(a_i), q_t(a_i)$, respectively.
Moreover, we denote the number of states of the finite automaton governing the behavior of the three agents by $N$.

In our lower bound proof, we show for each finite automaton that three agents governed by this automaton are not sufficient to explore the grid (or, more precisely, that there is an adversarial schedule for this automaton under which the agents do not explore every cell of the grid).
In this context, we consider the number $N$ as a constant, which also implies that the result of applying any fixed polynomial function to $N$ is a constant as well.
For the proof of our lower bound we require another intuitive definition.
Let $\ell$ be an infinite line in the Euclidian plane and $d$ some positive real number.
Let $B$ be the set of all points in the plane with integer coordinates and Euclidian distance at most $d$ to $\ell$.
Let $B'$ be the set of all grid cells that have the same coordinates as some point in $B$. 
Then we call $B'$ a \emph{band}.

\subsection*{A Single Agent}

Consider a single agent $a$ moving on the grid.
Since the number of states of its finite automaton is finite, $a$ must repeat a state at some point, i.e., there must be points in time $t, t'$ such that $q_t(a) = q_{t'}(a)$ and $q_{t''}(a) \neq q_t(a)$ for all $t < t'' < t'$.
As shown in \cite{Emek15}, agent $a$ will then, starting at time $t'$, repeat the exact behavior it showed starting at time $t$ regarding both movement on the grid and updating of its state.
We call the $2$-dimensional vector $c_{t'}(a) - c_t(a) = (x_{t'}(a) - x_{t}(a), y_{t'}(a) - y_{t}(a))$ the \emph{travel vector} of agent $a$ (from time $t$ to time $t'$).
Moreover, we call the time difference $t' - t$ the \emph{travel period}.

Note that travel vector and travel period do not depend on the choice of $t$ and $t'$ (provided $t$ and $t'$ satisfy the properties mentioned above).
In the case of multiple agents, we use the same definitions for any time segment where only a single agent is scheduled and does not encounter another agent.
In particular, we can only speak of a travel vector and a travel period when there are two points in time (in the considered time segment) where the scheduled agent repeats a state and at both times as well as in the time between, the agent is alone in its cell.

\section{Techniques} \label{sec:antstech}

In order to show our main result, we use a (large) proof by contradiction.
In the following we give a (very informal and possibly slightly inaccurate) high-level overview of how it proceeds. 
Our assumption, that holds throughout the remainder of the paper, is that three agents actually suffice to explore the grid.
From this assumption, we derive a contradiction as follows:

First, we fix an adversarial schedule for the three agents that has certain advantageous properties.
(We will show that it is already possible to derive a contradiction for this specific schedule.)
Then, using the finiteness of the number of configurations of agents in any bounded area, we show that for each distance $D$ there is a point in time such that from this time onwards, there are always at least two agents that have distance at least $D$.
However, since we can prove that any two agents must meet infinitely often, there must be infinitely many travels between the two far-away agents (which are not always the same agents).
We show that the vector along which such a travel takes place must have a fixed slope that is the same for all such travel vectors (from a sufficiently large point in time on).
Otherwise, there would exist two subsequent travels forth and back of different slope, which would imply that the traveling agent on its way back would miss the agent it is supposed to meet (which is the agent from whose position the first of the two travels started, roughly speaking).
This also holds if the traveling agent explores some area to the left and right of its travel direction (during its travel), since the distance $D$ between the two endpoints can be made arbitrarily large.

The crucial part of the proof is to show that the state of the traveling agent at the end of its travel does not depend on the exact vector between the start and the endpoint of its travel, but only on this vector ``modulo'' some other vector $v$ that is obtained by combining all of the finitely many possible traveling vectors of the aforementioned fixed slope.
Proving this statement enables us to show that, at the start of a travel, the information 1) about the states and relative locations ``modulo $v$'' of the agents, and 2) about which agent is scheduled next and which is the traveling agent, are sufficient to determine the same information at the start of the next travel.
Since there are only finitely many of these information tuples (exactly because they contain only the modulo version of the relative locations), at some point a tuple has to occur again.
Hence, in a sense, the whole configuration consisting of the three agents repeats its previous movement from this point on, at least if one ignores any movement in the direction of the fixed slope.
Thus, in each repetition between two occurrences of the information tuple, the whole configuration moves by some fixed (and always the same) vector, which implies that the agents explore ``at most half'' of the grid.

\section{The Schedule}

From this section on, we assume that three semi-synchronous agents whose behavior is governed by a finite automaton suffice to explore the grid.
Let $a_1$, $a_2$ and $a_3$ be these agents.
We start our proof by contradiction by specifying a schedule that we assume to be the adversarial schedule for the remainder of this paper:

We first schedule agent $a_1$ for some number of time steps, then agent $a_2$, then $a_3$, and then we iterate, again starting with $a_1$.
The number of steps an agent is scheduled can vary.
In other words, we can describe our schedule as a sequence 
\[
	\mathcal S = \left (\mathcal{S}_1^1, \mathcal{S}_1^2, \mathcal{S}_1^3, \mathcal{S}_2^1, \mathcal{S}_2^2, \mathcal{S}_2^3, \mathcal{S}_3^1, \dots \right )
\]
of subschedules where in each subschedule $\mathcal{S}_j^i$ only agent $a_i$ is scheduled.
The number of time steps in a subschedule $\mathcal{S}_j^i$ is determined as follows:

\begin{enumerate}

\item \label{item: meet} If there is a (finite) number $u > 0$ of time steps after which agent $a_i$ is in a cell occupied by another agent, then the subschedule $\mathcal{S}_j^i$ ends after $u_{\min}$ time steps where $u_{\min}$ denotes the smallest such $u$.

\item \label{item: repeat} If Case \ref{item: meet} does not apply, but there is a (finite) number $u > 0$ of time steps after which $a_i$ is in the same state in the same cell as it was at some earlier point in time during $\mathcal{S}_j^i$,
then do the following:

Fix a total order on the state space of $a_i$'s finite automaton.
(This total order can be chosen arbitrarily, but in each application of Case \ref{item: repeat} for agent $a_i$ the same order has to be used.)
Let $q$ be the smallest state according to this order which $a_i$ assumes at least twice in the same cell (if we scheduled $a_i$ indefinitely).
Then $\mathcal{S}_j^i$ ends after the smallest positive number of steps after which $a_i$ is in state $q$ and in a cell where $a_i$ would assume $q$ at least twice.
Note that the property that $a_i$ would assume $q$ twice implies that it would repeat the exact behavior between the first and the second assumption of $q$ infinitely often afterwards, thus iterating through the exact same movement on and on.

\item \label{item: escape} If none of the two above cases occurs, i.e., $a_i$ would move on indefinitely without meeting any other agent or being in the same state in the same cell as before, then we schedule as follows:
Let $(x,y)$ be the travel vector of $a_i$'s movement, and $k$ the travel period.
Then the subschedule $\mathcal{S}_j^i$ ends at the first time $t$ (strictly after the start of $\mathcal{S}_j^i$) for which the following property is satisfied:

For each cell $(x_t^r, y_t^r)$ occupied by an agent $a_r, r\neq i$, we have that 1) $x_t^i - x_t^r > k$ if $x > 0$, and $x_t^i - x_t^r < -k$ if $x < 0$, and 2) $y_t^i - y_t^r > k$ if $y > 0$, and $y_t^i - y_t^r < -k$ if $y < 0$.
The definition of the travel vector ensures that there is such a (finite) point in time $t$.
Note that Case \ref{item: escape} can only occur if $x \neq 0$ or $y \neq 0$.
Moreover, if this case actually occurs, then the complete subsequent schedule is adapted according to the following special rule (overriding all of the above): After time $t$, the two agents $a_r, r \neq i,$ are scheduled for one time step each (in arbitrary order), then agent $a_i$ is scheduled for $k$ time steps, i.e., exactly one travel period, and then we iterate this new scheduling.

\end{enumerate}

\begin{figure}

	\centering

	\begin{subfigure}{.45\textwidth}
		\centering
		\includegraphics{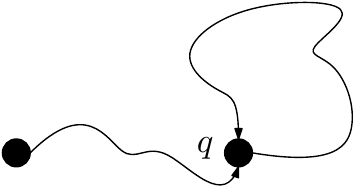}
		\caption{}\label{fig:schedule1}
	\end{subfigure}
	\begin{subfigure}{.45\textwidth}
		\centering
		\includegraphics{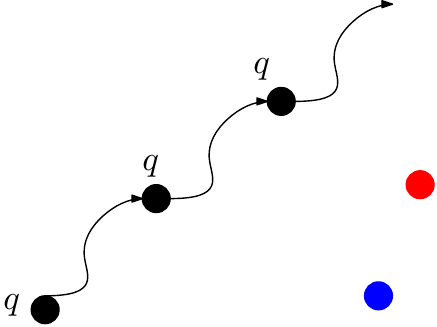}
		\caption{}\label{fig:schedule2}
	\end{subfigure}

	\caption{In Figure~\ref{fig:schedule1}, Case \ref{item: repeat} of our schedule is shown. Note that the agent already stops when it visits the cell on the right (in state $q$) for the first time (unless this happens after $0$ time steps). In Figure~\ref{fig:schedule2}, we see Case \ref{item: escape} of our schedule. One agent would move arbitrarily far away if scheduled sufficiently long. By letting this agent move away far enough and then scheduling it sufficiently often for a long enough period of time, we make sure that it will not interact anymore with any of the other two agents.}

	\label{fig:schedule}

\end{figure}

\noindent Observe that according to this schedule, the number of time steps a scheduled agent can stay put in a cell during one of its subschedules is upper bounded by $N$.
Also note that in each of the three cases, the number of steps in the subschedule is positive (and finite).
For an illustration of Cases \ref{item: repeat} and \ref{item: escape}, see Figure~\ref{fig:schedule}.
We now collect a few lemmas that highlight certain properties of the three cases.

\begin{lemma} 

	\label{lemma: noescape}

	Case \ref{item: escape} cannot occur.

\end{lemma}

\begin{proof}
	Recall that we assume (globally) that the three agents explore the entire infinite grid.
	Assume that Case \ref{item: escape} occurs and let $a_i$ denote the agent that would move on indefinitely without meeting another agent.
	Then, at the beginning of the first iteration according to the special rule, the distance of agent $a_i$ to any of the other agents is more than $k$ in at least one (of $x$- and $y$-) direction and $a_i$ moves away from the agents according to the travel vector.
	After each of the other agents makes a step, this distance is still at least $k$.
	Hence, agent $a_i$ cannot encounter one of the other agents during its next $k$ steps, since in total it moves away from the other agents, according to the specification of Case \ref{item: escape}.

	The direction of the travel vector also ensures that the distance to the other agents is again increased to more than $k$ (in at least one direction).
	Thus, the same arguments hold for the next iteration, and we obtain by induction that agent $a_i$ will never encounter another agent after the occurrence of Case \ref{item: escape}.
	It follows that, if three agents suffice to explore the grid, then also a team of two agents and a separate single agent can explore the grid without any communication between the team and the single agent.
	From~\cite{Emek15}, we know that this is not possible since a team of two agents (hence, also a single agent) can only explore a band of constant width.
\end{proof}

Following Lemma \ref{lemma: noescape}, we will assume in the following that Case \ref{item: escape} does not occur, i.e., each agent's subschedule ends because it encounters another agent or because it repeats a pair state/cell.
This allows us to group the possible subschedules of an agent into two categories:
We say that a subschedule $\mathcal{S}_j^i$ is of \emph{type 1} if $\mathcal{S}_j^i$ ends because of the condition given in Case \ref{item: meet}, and of \emph{type 2} if $\mathcal{S}_j^i$ ends because of the condition given in Case \ref{item: repeat}.

\begin{lemma}

	\label{lemma: type2}

	Any subschedule of type 2 consists of at most $N$ time steps.

\end{lemma}

\begin{proof}

	Assume for a contradiction that there is a subschedule $\mathcal{S}_j^i$ of type 2 that consists of at least $N + 1$ time steps and starts at some time $t$.
	Then, by the pigeonhole principle, there must be two points in time $t < t' < t'' \leq t + N + 1$ such that $q_{t'}^i = q_{t''}^i$.
	Moreover, it must also hold that $c_{t'}^i = c_{t''}^i$ since otherwise $a_i$ would move according to some non-zero travel vector (from time $t'$ onwards) which would imply that $\mathcal{S}_j^i$ is not of type 2.

	This implies that if $a_i$'s subschedule would also continue at and after time $t + N + 1$ on an empty grid, then $a_i$ would cycle through the same movement on and on, starting from time $t'$.
	Hence, if there is a cell $c$ that is visited by $a_i$ in some state $q$ in the (continued) movement after time $t''$, then there must also be a point in time before $t''$ (during $\mathcal{S}_j^i$) at which $a_i$ visits $c$ in state $q$.
	It follows from the definition of our schedule that $\mathcal{S}_j^i$ ends before time $t''$, yielding a contradiction to our assumption.

\end{proof}

\begin{lemma}
	\label{lemma: type1}
	Any subschedule $\mathcal{S}_j^i$ of type 1, where agent $a_i$ ends in the same cell from which it started, consists of at most $N (2N + 1)$ time steps.
	More generally, any subschedule $\mathcal{S}_j^i$ of type 1, where $a_i$ ends in a cell of distance at most $D$ from the cell from which it started, consists of at most $N (2N + 1 + D)$ time steps.
\end{lemma}

\begin{proof}
	We start by proving the special case where $a_i$ ends in the same cell from which it started.
	Suppose for a contradiction that there is a subschedule $\mathcal{S}_j^i$ as described in the lemma that consists of more than $N (2N + 1)$ time steps.
	Let $t$ and $u$ denote the points in time when $\mathcal{S}_j^i$ starts and ends, respectively.
	Since $a_i$ does not encounter any other agent between time $t$ and time $u$, it behaves like a single agent on an empty grid between $t$ and $u$.
	In particular, there is a travel vector $(x,y)$ of agent $a_i$ from time $t+1$ to time $u-1$ since $N (2N + 1) - 1 > N$.

	For reasons of symmetry, we can assume w.l.o.g.\ that $x>0$ and $y \geq 0$.
	Note that $x = 0 = y$ is not possible since in that case $a_i$ would cycle through the same (cyclic) movement over and over without meeting any other agent, which would imply that $\mathcal{S}_j^i$ is not of type 1.
	Let $p$ be the travel period which, according to its definition, is at most $N$.
	Let $q$ be the state whose second occurrence during $\mathcal{S}_j^i$ (excluding the occurrence of the state at the beginning of $\mathcal{S}_j^i$) comes earliest.
	Let $t'$ be the time when $q$ occurs for the first time.
	Since $t' \leq t + N$, we know that $x_{t'}^i \geq x_t^i - N$.

	Now, as in each travel period $a_i$ increases the $x$-coordinate of the cell it occupies by at least $1$, it follows that at time $t' + 2N \cdot p$ the $x$-coordinate of the cell $a_i$ occupies is at least $x_t^i + N$.
	Furthermore, since in each further travel period agent $a_i$ would advance by at least one cell in (positive) $x$-direction in total and $p \leq N$, after time $t' + 2N \cdot p$ agent $a_i$ will never have an $x$-coordinate of less than $x_t^i + 1$, i.e., it will never reach $c_t^i$ then.
	But $a_i$ also cannot have visited $c_t^i (= c_u^i)$ between time $t+1$ and $t' + 2N \cdot p$ since $t' + 2N \cdot p \leq t + N (2N + 1)$ and we assumed that $\mathcal{S}_j^i$ consists of more than $N (2N + 1)$ time steps.
	Thus, we obtain a contradiction, which proves the first lemma statement.	

	For the more general second statement, by an analogous proof we obtain that after time $t' + 2N \cdot p + D \cdot p$ agent $a_i$ will never have an $x$-coordinate of less than $x_t^i + 1 + D$, i.e., it will never reach $c_u^i$ then.
	But, since $t' + 2N \cdot p + D \cdot p \leq t + N (2N + 1 + D)$, $a_i$ also cannot have visited $c_u^i$ between time $t+1$ and $t' + 2N \cdot p + D \cdot p$, under the assumption that $\mathcal{S}_j^i$ consists of more than $N (2N + 1 + D)$ time steps.
	Hence, this assumption must be false, and the lemma statement follows.
\end{proof}

\section{Traveling and Meeting}

Having defined and studied the schedule, we now proceed with our lower bound proof as described in Section \ref{sec:antstech}.
The next lemma shows that for each distance there is a point in time after which the farthest two agents are never closer than this distance. 

\begin{lemma}
	\label{lemma: distance}
	For each distance $D$ there is a time $T$ such that at any time $t \geq T$ the largest pairwise distance of the three agents is at least $D$.
\end{lemma}
\begin{proof}
	Suppose that the lemma statement is not true.
	Then there is an infinite sequence $\mathcal T$ of points in time such that at each of these points in time the largest pairwise distance of the three agents is less than $D$.
	Since the distances of the agents are less than $D$ at all points in time from $\mathcal T$ and the number of states the three agents can be in is finite, it follows that there must be points in time $t, t' \in \mathcal T$ such that 1) each agent is in the same state at $t$ and $t'$, 2) $x_t^i - x_t^j = x_{t'}^i - x_{t'}^j$ and $y_t^i - y_t^j = y_{t'}^i - y_{t'}^j$ for all $i, j \in \{ 1, 2, 3 \}, i \neq j$, and 3) the same agent is scheduled to move next.
	Since the agents are oblivious of the absolute coordinates of the grid, this implies that from time $t'$ on, the agents will repeat the exact behavior they showed starting at time $t$.
	(Note that we use here that the schedule following a configuration is uniquely determined by the above information.)
	Hence, at time $t' + (t' - t)$ the agents will again be in the exact same configuration and so on.

	Define $(x, y) = (x_{t'}^i - x_t^i, y_{t'}^i - y_t^i)$, where $i = 1$ (which implies that this equation also holds for $i = 2, 3$).
	Vector $(x, y)$ describes the total movement of each of the agents during each of the (repeating) time periods of length $t' - t$.
	It follows that each cell that has not been explored by time $t$ must be at distance at most $t' - t$ from some cell that is obtained by adding a multiple of the vector $(x,y)$ to one cell from $\{ c_t^1, c_t^2, c_t^3 \}$; otherwise it will never be explored.
	Since each such cell at distance at most $t' - t$ (which is constant) must lie in a band of constant width and ``direction'' $(x,y)$ that contains $c_t^1$, $c_t^2$ or $c_t^3$, there are infinitely many cells that must have been explored before time $t$.
	This yields a contradiction.
\end{proof}

\begin{figure}
	\centering
		\centering
		\includegraphics{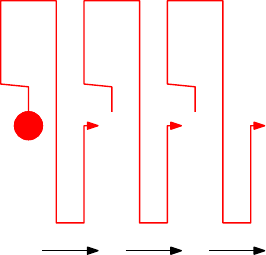}

 	\caption{An example showing a possible movement (red) of an agent whose travel vector is given by the black arrows. The agent performs the total movement given by the travel vector in at most $N$ time steps, or more precisely, during one travel period. 
	}
	\label{fig:travelVector}

\end{figure}

For any distance $D$, we denote by $T_D$ the smallest time $T$ for which it holds that at any time $t \geq T$ the largest pairwise distance of the three agents is at least $D$.
In the following we collect a number of useful definitions regarding the meetings of different agents.
In particular, we distinguish between three different types of agents at times when one agent is traveling from another agent to the far-away agent whose existence is certified by Lemma \ref{lemma: distance}.
For an illustration of how a large distance between agents influences choices of travel vectors, see Figure~\ref{fig:travelVector}.

\begin{definition}

	For any $t \geq 0$, we define the \emph{meeting set} $M_t$ as the set of agents that are not alone in the cell they occupy, at time $t$.
	We call the infinite sequence $(M_0, M_1, \dots)$ the \emph{meeting sequence}.
	If for a subsequence $(M_t, M_{t+1}, \dots, M_{t + i})$ of the meeting sequence it holds that $i>0$, $M_t \neq \emptyset \neq M_{t + i}$ and $M_{t + j} = \emptyset$ for all $0 < j < i$, then we call the pair $(t, t + i)$ a \emph{meeting pair}.
	Now, let $(t, u)$ be a meeting pair such that $\vert M_t \vert = 2 = \vert M_u \vert$ and $M_t \neq M_u$.
	Then we call $(t, u)$ a \emph{travel meeting pair}.
	Moreover, we call the (uniquely defined) agent $a$  contained in $M_t \cap M_u$ a \emph{traveling agent (for $(t, u)$)}, the agent contained in $M_t \setminus \{ a \}$ a \emph{source agent} and the agent contained in $M_u \setminus \{ a \}$ a \emph{destination agent}.

\end{definition}

In order to continue according to our high-level proof idea from Section \ref{sec:antstech}, we need a few helping lemmas that highlight properties of the previous definitions.
We start with a lemma that shows an important property of the meeting sequence:

\begin{lemma}

	\label{lemma: meeting}

	Each of the three agents is contained in infinitely many of the $M_t$ from the meeting sequence.

\end{lemma}

\begin{proof}

	Suppose that there is an agent $a_i$ that is not contained in infinitely many of the $M_t$, i.e., there is a point in time $u$ such that $a_i \notin M_t$ for all $t \geq u$.
	Then, starting from time $u$, the exploration by the two agents $a_r, r \neq i$ is entirely independent of the exploration by agent $a_i$ since they never meet again.
	Thus, we get a contradiction analogously to the argumentation in the proof of Lemma \ref{lemma: noescape}.
\end{proof}

Next, we study travel meeting pairs more closely.
In Lemma \ref{lemma: once}, we present bounds on the number of subschedules of the different types of agents in the time frame given by a travel meeting pair, and examine the types of the subschedules.
Afterwards, in Lemma \ref{lemma: betweentravels}, we bound the number of time steps between two subsequent travel meeting pairs from above.
In both cases, the results only hold from a large enough point in time onwards, but this is sufficient for our purposes since before that point in time only a constant number of cells were explored.
Note that, in general, we do not attempt to minimize the dependence on $N$ in our bounds as showing the finiteness of certain parameters is, again, sufficient for our purposes.
Instead we prefer to choose the simplest arguments that lead to the desired finiteness results, even if they augment the actual bound by a few factors of $N$.

\begin{lemma}

	\label{lemma: once}

	There is a point in time $T$ such that, for each travel meeting pair $(t, u)$ with $t \geq T$, the following properties hold:

	\begin{enumerate}

		\item The traveling agent for $(t, u)$ is scheduled exactly once (for a number of time steps) between time $t$ and time $u$.

		\item The subschedule of the traveling agent is of type $1$ and ends exactly at time $u$.

		\item The source and the destination agent for $(t,u)$ are scheduled at most once (for a number of time steps).

		\item If the source or the destination agent is scheduled, then its subschedule is of type 2.

	\end{enumerate}

\end{lemma}

\begin{proof}

	Recall the definition of $T_D$ for any distance $D$.
	Let $T \geq T_{2N + 1}$, and consider an arbitrary travel meeting pair $(t, u)$ with $t \geq T$ and traveling agent $a_i$.
	Observe that if the source agent is scheduled between time $t$ and time $u$, then its subschedules must be of type 2, because the source agent is not contained in the meeting set $M_u$.
	Hence, if $a_i$ is not scheduled at all between time $t$ and time $u$, then the source agent must be scheduled at most once (because of the specification of our schedule) which implies that its distance from $c_t^i$ at time $u$ is at most $N$, by Lemma \ref{lemma: type2}.
	But since in this case $a_i$ and the destination agent meet at $c_t^i$ at time $u$, we obtain a contradiction to the fact that $T \geq T_{2N + 1}$.
	Thus, we know that $a_i$ is scheduled at least once between time $t$ and time $u$.	

	Now, assume for a contradiction that the first subschedule of $a_i$ between time $t$ and time $u$ is of type 2.
	This implies that if one would schedule $a_i$ on and on, it would repeat a state in the same (empty) cell after at most $N + 1$ time steps and then cycle through (a part of) the same movement it performed before.
	Hence, even if there are more subschedules for $a_i$ than one (between time $t$ and time $u$), it will never reach a cell that has a distance of more than $N$ from $c_t^i$.
	Since analogous statements hold for the source agent, we know that at time $u$ the distance between the source agent and the cell where $a_i$ and the destination agent meet is at most $2N$ which again contradicts our specification of $T$.
	Thus, we know that the first subschedule of $a_i$ is of type 1.

	It follows that $a_i$'s subschedule ends exactly at time $u$ since the subschedule must end with $a_i$ meeting the destination agent, which also implies that $a_i$ is scheduled exactly once between time $t$ and time $u$.
	Moreover, the subschedules of the source and the destination agent (if they are scheduled at all between time $t$ and time $u$) must be of type 2 since $(t,u)$ is a (travel) meeting pair.
	Furthermore, by the nature of our schedule, the source and the destination agent must be scheduled at most once between time $t$ and time $u$.
\end{proof}

\begin{lemma}
	There is a point in time $T$ such that the following holds:
	If $(t,u)$ and $(t', u')$ are travel meeting pairs such that $T \leq t < t'$ and there exists no travel meeting pair $(t'', u'')$ with $t < t'' < t'$, then $t' - u \leq 8 (N + 1)^5$.
	\label{lemma: betweentravels}
\end{lemma}
\begin{proof}
	Observe that from the definition of a travel meeting pair it follows that $t' \geq u$.
	Set $T := T_{N(4N + 1)+1}$ and let $t, u, t', u'$ be as described in the lemma.
	W.l.o.g., let $a_1$ and $a_2$ be the agents contained in $M_u$.
	Let $t_1 < t_2 < \dots < t_k$ be exactly the points in time $t_j$ between $u$ and $t'$
	for which $M_{t_j} \neq \emptyset$ holds.
	It follows that all $M_{t_j}$ are identical to $M_u = M_{t'}$.

	We claim that $k < (2 N^2 + 1)(N + 1)$.
	Suppose for a contradiction that $k \geq (2 N^2 + 1)(N + 1)$.
	Then, there must be at least $2 N^2 + 1$ indices $g \in \{ 1, \dots, k \}$ such that $a_1$ or $a_2$ are scheduled to move at time $t_g$ since each subschedule of $a_3$ between time $u$ and $t'$ is of type 2 and hence consists of at most $N$ time steps, by Lemma \ref{lemma: type2}.
	It follows that there must be some $1 \leq g < h \leq k$ such that $a_1$ and $a_2$ are in the exact same (pair of) states at time $t_g$ and at time $t_h$, and the same agent is scheduled next.
	This implies that $a_1$ and $a_2$ go through the same movement that they executed between time $t_g$ and $t_h$, over and over again, starting from time $t_h$, until at least one of them encounters agent $a_3$.
	Recall that our schedule is oblivious, i.e., only depends on the current states and locations of the agents and the information which agent is scheduled next.

	During this movement, i.e., anytime between $t_g$ and $u'$, our agents $a_1$ and $a_2$ cannot move too far from each other as we show in the following:
	If an agent executes a subschedule of type 2 and then its next subschedule is again of type 2 (and no other agent is in the cell from which this second subschedule starts), then the agent ends both subschedules in the same cell and the same state, due to the specification of type 2 subschedules.
	Hence, if $a_1$ and $a_2$ together perform three consecutive subschedules of type 2 (disregarding any subschedules of $a_3$), then in each subsequent subschedule they will repeat the same movement as in the last subschedule, until one of them encounters $a_3$.\footnote{This statement does not necessarily hold for only two consecutive subschedules of type 2 since the agent $a$ who moved first may encounter the second agent (who moved during the subsequent subschedule) during $a$'s next subschedule.}
	Since we already established that $a_1$ and $a_2$ would repeat the movement they executed between time $t_g$ and $t_h$ (which includes subschedules of type 1) over and over again if none of them met $a_3$, it cannot be the case that there are three consecutive subschedules of type 2 (of $a_1$ and $a_2$).
	Hence, between any two subschedules (of the agents $a_1$ and $a_2$) of type 1 between time $t_g$ and time $u'$, there are at most two subschedules of type 2 (of those agents).
	Now by Lemma \ref{lemma: type2} and Lemma \ref{lemma: type1}, the fact that after each subschedule of type 1 the agents $a_1$ and $a_2$ are in the same cell, and our established observation about the cyclic movement following $t_h$, it follows that the maximum distance of the two agents between time $t_g$ and time $u'$ is at most $N(2N + 1 + D)$ where $D = 2N$.
	Hence, when one of the two agents encounters agent $a_3$, then the other is at distance at most $N(4N + 1)$.
	This contradicts the fact that $t \geq  T_{N(4N + 1)+1}$ and proves the claim.

	From the above, we obtain the following picture:
	There are at most $k < (2 N^2 + 1)(N + 1)$ subschedules of type 1 between time $u$ and $t'$ (since, when a subschedule of type 1 ends, the corresponding element from the meeting sequence is non-empty).
	Between any two subschedules of type 1 (and possibly before the first/after the last) there are at most two subschedules of type 2 of agents $a_1$ and $a_2$, which gives us a total of at most $2 \cdot (2 N^2 + 1)(N + 1)$ subschedules of type 2 of agents $a_1$ and $a_2$ together between time $u$ and $t'$.
	For agent $a_3$, we obtain an upper bound of $1/2 \cdot 3 \cdot (2 N^2 + 1)(N + 1) + 1$ for the number of subschedules between time $u$ and $t'$ (which are all of type 2).
	Now, by Lemma \ref{lemma: type2} and Lemma \ref{lemma: type1}, we obtain that $t' - u \leq (7/2 \cdot (2 N^2 + 1)(N + 1) + 1)N + (2 N^2 + 1)(N + 1)N(2N + 1 + D)$ where $D = 2N$.
	Hence, $t' - u \leq 8 (N + 1)^5$.
\end{proof}	

Using Lemma \ref{lemma: betweentravels}, we show in the following that for any travel meeting pair $(t, u)$, the information about the states of the agents, which two agents are in the same cell, and who is scheduled next, all at time $u$, already uniquely determines a lot of information about the agents at the starting time of the next travel meeting pair.
Again, this result only holds from a sufficiently large point in time onwards.
This concludes our collection of helping lemmas.

\begin{lemma}
	There is a point in time $T$ such that the following holds:
	For any two subsequent travel meeting pairs $(t, u)$, $(t', u')$ with $T \leq t < t'$, the tuple $(q_u^1, q_u^2, q_u^3, \anext{u}, M_u)$ uniquely determines the tuple $(q_{t'}^1, q_{t'}^2, q_{t'}^3, c_{t'}^1 - c_u^1, c_{t'}^2 - c_u^2, c_{t'}^3 - c_u^3, \anext{t'}, M_{t'})$, where $\anext{u}$, resp. $\anext{t'}$, denotes the agent that is scheduled at time $u$, resp.\ $t'$.
	\label{lemma: traveltravel}
\end{lemma}
\begin{proof}
	Let $T$ be sufficiently large so that $T \geq T_{8(N+1)^5 + 2N + 1}$ holds and Lemma \ref{lemma: once} and Lemma \ref{lemma: betweentravels} apply.
	Let $t$, $u$, $t'$, $u'$ be as described in the lemma. 
	Observe that the subschedule of the traveling agent for $(t,u)$ ends exactly at time $u$, by Lemma \ref{lemma: once}, and thus the subschedule of $\anext{u}$ actually starts at $u$.
	Moreover, due to the choice of $T$, the agent not contained in $M_u$ will not be in the same cell as another agent until at least (and including) time $u + 8(N+1)^5 + 2N$, while the two agents contained in $M_u$ are in the same cell at time $u$.
	Hence, if we knew of which types the subschedules of the three agents are until time $u + 8(N+1)^5$, then we could (deterministically) compute the exact behavior of the three agents up to time $u + 8(N+1)^5$.	

	Fortunately, the types of the subschedules are uniquely determined, in a simple way, by the tuple $(q_u^1, q_u^2, q_u^3, \anext{u}, M_u)$.
	Consider the first subschedule, i.e., the one of agent $\anext{u}$ starting at time $u$.
	Now consider the state of $\anext{u}$ whose second occurrence would come earliest if we scheduled $\anext{u}$ indefinitely (not counting the occurrence of a state at time $u$ and ignoring the existence of any other agents on the grid).
	If this second occurrence happens in the same cell as the previous occurrence of the same state and $\anext{u}$ did not encounter any other agent after time $u$ until the time of the second occurrence, then the subschedule of $\anext{u}$ must be of type \ref{item: repeat} according to the specification of our schedule.
	Otherwise, it must be of type \ref{item: meet}.	

	Note that $\anext{u}$ could only have encountered another agent until the second occurrence if these two agents are contained in $M_u$, due to our choice of $T$.
	More generally, any meeting during such a ``simulation'' (for determining if the respective subschedule is of type \ref{item: repeat}) must be between the two agents from $M_u$ if the agent whose potential subschedule is simulated starts its subschedule at time $u + 8(N+1)^5$ at the latest, again due to our choice of $T$.
	(Here, we use that such a simulation contains at most $2N$ time steps until the second occurrence.)	

	Now, the specification of the type of the subschedule in combination with the information contained in the tuple $(q_u^1, q_u^2, q_u^3, \anext{u}, M_u)$ uniquely determines the states of the agents at the time the subschedule of $\anext{u}$ ends, their relative locations compared to time $u$ and which agent is to move next.
	Then, we can iterate this argument for the second, third, $\dots$, subschedule (from time $u$ on) and obtain that for each of these subschedules the exact movement of the scheduled agent is uniquely determined by $(q_u^1, q_u^2, q_u^3, \anext{u}, M_u)$.
	Again, this argumentation holds up to (and including) time $u + 8(N+1)^5$.
	By Lemma \ref{lemma: betweentravels}, we know that $t' \leq u + 8(N+1)^5$.
	On the other hand, we have that $u' > u + 8(N+1)^5$, due to our choice of $T$.
	Hence, the agent scheduled at time $u + 8(N+1)^5$ must be the traveling agent for $(t', u')$.

	Moreover, the above considerations ensure that the exact behavior of the agents up to (and including) time $u + 8(N+1)^5$ is uniquely determined by $(q_u^1, q_u^2, q_u^3, \anext{u}, M_u)$.
	Therefore, also the traveling agent for $(t', u')$ is uniquely determined by $(q_u^1, q_u^2, q_u^3, \anext{u}, M_u)$, and also the parameter $t' - u$.
	Since $t' \leq u + 8(N+1)^5$, it follows that $(q_u^1, q_u^2, q_u^3, \anext{u}, M_u)$ uniquely determines the tuple $(q_{t'}^1, q_{t'}^2, q_{t'}^3, c_{t'}^1 - c_u^1, c_{t'}^2 - c_u^2, c_{t'}^3 - c_u^3, \anext{t'}, M_{t'})$.
\end{proof}

\section{The Travel Vector and a Modulo Operation}

After collecting the above helping lemmas, we are now all set to formally prove the (remaining) statements from our proof sketch.
Before going through the statements one by one, let us for convenience define the notion of a travel:
Let $(t,u)$ be a travel meeting pair.
By Lemma \ref{lemma: once}, we know that the traveling agent for $(t,u)$ is scheduled exactly once between $t$ and $u$.
We call the corresponding subschedule (or the movement during that subschedule) a \emph{travel}. 
Recall the definition of travel vector and travel period.
Note that a travel only has a travel vector (and period) if the traveling agent repeats a state (in empty cells) during the travel.
Furthermore, observe that if a travel has a travel vector, then at least one entry of the travel vector is non-zero, due to the choice of our schedule.
We now prove the first of the remaining statements, namely, that after a certain point in time, any travel vector has the same slope.

\begin{lemma}

	There is a point in time $T$ and a (possibly negative) ratio $r$ such that each travel starting at time $T$ or later has travel vector $(x,y)$ with $y/x = r$.
	For the sake of simplicity, assume that $r$ is set to $\infty$ if $x = 0$.

	\label{lemma: vector}

\end{lemma}

\begin{proof}

	Let $T$ be sufficiently large so that $T \geq T_{N+2}$ holds and Lemma \ref{lemma: once} and Lemma \ref{lemma: betweentravels} apply.
	Then we know that any travel starting at time $T$ or later actually has a travel vector (and period).
	Now, consider two travel meeting pairs $(t,u)$ and $(t', u')$ with $T \leq t < t'$ such that there is no travel meeting pair $(t'', u'')$ with $t < t'' < t'$.
	Let $(x,y), (x', y')$ be the travel vectors for the travels corresponding to $(t,u)$ and $(t', u')$, respectively.
	Assume that $y'/x' \neq y/x$, where, again, we set the ratio to $\infty$ if the denominator is $0$.
	Note that not both of $x$ and $y$ (or $x'$ and $y'$) can be $0$.
	Let $c_0$ and $c_1$ be the cells at which the travel with travel vector $(x,y)$ starts and ends, respectively, and $c'_0$ and $c'_1$ analogously for the travel with travel vector $(x',y')$.	

	By the characterization of the travel of a single agent and the fact that the travel period is always at most $N$, we know that there are positive integers $b$ and $b'$ such that $\distance{c_1}{c_0 + b \cdot (x,y)} \leq N$ and $\distance{c'_1}{c'_0 + b' \cdot (x', y')} \leq N$.
	Moreover, by Lemma \ref{lemma: type2} and Lemma \ref{lemma: once}, the source agent for $(t,u)$ travels at most a distance of $N$ between time $t$ and $u$ since its subschedule is of type 2 if the agent is scheduled at all.
	The same holds for the destination agent for $(t', u')$ between time $t'$ and $u'$.
	By Lemma \ref{lemma: betweentravels}, it follows that $\distance{c_0}{c'_1} \leq 8 (N + 1)^5 + 2 N$ (since the source agent for the first of the two travels is the destination agent for the second) and $\distance{c_1}{c'_0} \leq 8 (N + 1)^5$.
	Combining our above distance observations, we also obtain $\distance{c'_1}{c_0 + b \cdot (x,y) + b' \cdot (x',y')} \leq N + 8 (N + 1)^5 + N$, which together with $\distance{c_0}{c'_1} \leq 8 (N + 1)^5 + 2 N$ implies $\distance{c_0}{c_0 + b \cdot (x,y) + b' \cdot (x',y')} \leq 16 (N + 1)^5 + 4 N$.

	Let $D \geq N$ be some positive integer.
	We now require, additionally to the above requirements regarding $T$, that $T \geq T_D$.
	Also fix some arbitrary $x, y, x', y'$ such that $(x,y)$ and $(x', y')$ are possible travel vectors of a single agent.
	For a contradiction, assume that $x, y, x', y'$ have the properties specified at the beginning of the proof (which implies that also all of the above conclusions hold).	

	At the time when the first of the two considered travels starts there are two agents at $c_0$ and $c_1$ while the last agent is in distance at most $N$ from $c_0$.
	Hence, the distance between $c_0$ and $c_1$ is at least $D - N$.
	This implies that $b \cdot (|x| + |y|) \geq \distance{c_1}{c_0} - N \geq D - 2 N$.
	Analogously, we obtain $b' \cdot (|x'| + |y'|) \geq D - 2 N$.
	Since $x, y, x', y'$ are fixed, we can therefore make $b$ and $b'$ arbitrarily large by increasing $D$.
	By increasing $b$ and $b'$, we can in turn make $\distance{c_0}{c_0 + b \cdot (x,y) + b' \cdot (x',y')}$ arbitrarily large, since $y'/x' \neq y/x$ (which implies that there is an angle between the two vectors $(x,y)$ and $(x', y')$ that is not $0^{\circ}$ or $180^{\circ}$).
	Hence, if $D$ is sufficiently large, then the above inequality $\distance{c_0}{c_0 + b \cdot (x,y) + b' \cdot (x',y')} \leq 16 (N + 1)^5 + 4 N$ is not satisfied anymore, which shows that $y'/x' = y/x$.

	Note that the magnitude $D$ has to reach for this (in our proof by contradiction) depends on $x, y, x', y'$.
	However, since the number of possible travel vectors of a single agent is bounded by the number of states in its finite automaton, we can simply derive a sufficiently large $D$ for each of the finitely many possible combinations for $x, y, x', y'$ and then choose a $T$ that is larger than all of the $T_D$.	
\end{proof}

Note that the exact value of $r$ depends only on the finite automaton governing the behavior of the three agents.
From now on, we denote the ratio whose existence is certified by Lemma \ref{lemma: vector} by $r$.
W.l.o.g., we can (and will) assume that $r \geq 0$ (and that $r \neq \infty$), for reasons of symmetry.
Recall that any travel vector has at least one non-zero entry.
The next step on our agenda is essentially to show that the state of an agent at the end of a travel does not depend on (the full information about) the vector between start and endpoint of that travel (and other parameters), but only on a reduced amount of information regarding this vector (and the other parameters).
More specifically, the required information about this vector is the result of applying a certain modulo operation to the vector.

We then proceed by showing that the information about 1) the states of the agents, 2) their relative locations after applying the modulo operation, 3) which agents shared a cell most recently, and 4) which agent is scheduled next, at the start of a travel, is enough to determine the exact same information at the end of the travel.
Now, we benefit from the previous reduction of information due to our modulo operation in the sense that we can show that there are only constantly many combinations of relative locations of the three agents (that can actually occur) after applying the modulo operation.
This, in turn, implies that there are only constantly many possibilities for the whole aforementioned information tuple at the start and end of a travel, which will enable us to prove our main theorem.
We start by defining our modulo operation in Definition \ref{def: modulo}.
Then we show a technical helping lemma, Lemma \ref{lemma: tuple}, which finally enables us to prove the aforementioned relation between the information tuple at the start and end of a travel in Lemma \ref{lemma: modulo}.
Note that for technical reasons, Lemma \ref{lemma: modulo} gives a slightly different statement than indicated above, dealing with travel meeting pairs instead of travels.

\begin{definition}

	\label{def: modulo}

Let $\{(x_1, y_1), (x_2, y_2), \dots, (x_k, y_k) \}$ be the set of travel vectors that the agents can have if you let one of them explore the grid starting in an arbitrary state (which clearly is a superset of the actually occurring travel vectors in our multi-agent case).
	Let $R$ be the subset of the above set that contains exactly the vectors $(x_j, y_j)$ that satisfy $y_j/x_j = r$.
	From now on, denote by $x$ the least common multiple of the $|x_j|$ from the vectors in $R$ and set $y := rx$.
	It follows that $(x,y)$ is a (possibly negative) integer multiple of any of the vectors from $R$.
	Note that $R$ cannot be empty since otherwise it is not possible that the agents explore the entire grid, due to Lemma \ref{lemma: meeting} and Lemma \ref{lemma: vector}.	

	Now, let $w, z$ be integers and let $b$ be the smallest integer such that $w+bx \geq 0$.
	(This is well-defined since $x > 0$, due to $r \neq \infty$.)
	We define $(w, z) \pmod{(x, y)} := (w + bx, z + by)$.
	For two cells $(w', z')$, $(w'', z'')$, we define $(w'', z'') \ominus (w', z') := (w'' - w', z'' - z') \pmod{(x, y)}$.

\end{definition}

Note that Definition \ref{def: modulo} ensures that for any $(w,z)$, $(w',z')$ where $(w' - w, z' - z)$ is a multiple of $(x,y)$, we have that $(w, z) \pmod{(x, y)} = (w', z') \pmod{(x, y)}$.

\begin{lemma}

	Let $a$ be an agent, $q$ a state from $a$'s finite automaton and $c, c', c''$ cells of the grid such that the following properties are satisfied:

	\begin{enumerate}

		\item $\distance{c}{c'} \geq N$ and $\distance{c''}{c'} \geq N$

		\item \label{prop: integer} There is an integer $b$ such that $c'' - c = b \cdot (w,z)$, where $(w,z)$ is agent $a$'s travel vector if it starts in state $q$.

		\item If agent $a$ starts in cell $c$ in state $q$ on an otherwise empty grid, then it arrives at $c'$ after finite time.

		\item If agent $a$ starts in cell $c''$ in state $q$ on an otherwise empty grid, then it arrives at $c'$ after finite time.

	\end{enumerate}

	Let $q'$ denote the state in which $a$ arrives at $c'$ (for the first time) when starting from $c$ (in state $q$), and $q''$ the state in which $a$ arrives at $c'$ (for the first time) when starting from $c''$ (in state $q$).
	Then it holds that $q' = q''$.

	\label{lemma: tuple}

\end{lemma}

\begin{proof}

	If $c = c''$, then the lemma holds trivially, thus assume that $c \neq c''$.
	W.l.o.g., we can assume that $b > 0$, which implies that, if agent $a$ starts in cell $c$ in state $q$ (say, at time $t$), then $a$ arrives at some point in time $u > t$ in cell $c''$ in state $q$ (possibly $a$ visited $c''$ before in some other state).
	Hence, if $a$ does not visit cell $c'$ between time $t$ and time $u$, then the lemma also holds since after arriving at $c''$ in state $q$, $a$ will perform the exact same movement as if it started in $c''$ in state $q$.

	Thus, consider the last remaining case, i.e., assume that $a$ visits $c'$ for the first time at some time $t < t' < u$.
	W.l.o.g., we can assume that $w$ and $z$ are non-negative and $w \geq z$.
	(Also recall that at least one of $w$ and $z$ is non-zero.)
	Let $c_0, c_1, \dots$ be the cells that $a$ visits in state $q$ at and after time $t$, where $c_0$ and $c_k$, for some $k > 0$, are the cells that $a$ visits at time $t$ and $u$, respectively, i.e., $c_0 = c$ and $c_k = c''$.
	Observe that $c_{j + 1} = c_j + (w,z)$ holds for each $j$.
	Denote the $x$-coordinates of $c'$ and $c_k = c''$ by $x'$ and $x''$, respectively.
	Since $w \geq z$, it follows that $\distance{c_j}{c'} \geq \distance{c''}{c'} \geq N$ for all $j \geq k$ if $x' \leq x''$, and $\distance{c_j}{c'} \geq \distance{c''}{c'} \geq N$ for all $0 \leq j \leq k$ if $x' \geq x''$.
	Let $h$ be the largest index such that $a$ visits $c_h$ in state $q$ at or before time $t'$.
	Then $h < k$, and $\distance{c_h}{c'} \leq N - 1$ since traveling from $c_h$ (in state $q$) to $c_{h+1}$ (in state $q$) takes $a$ at most one travel period, so at most $N$ time steps.
	If $x' \geq x''$, then we obtain a contradiction to our above observation, thus it follows that $x' < x''$.
	But this implies $\distance{c_j}{c'} \geq N$ for all $j \geq k$ which in turn implies for all $j \geq k$ that $c'$ cannot be visited by $a$ between visiting $c_j$ (in state $q$) and $c_{j+1}$ (in state $q$).
	Hence, $a$ does not visit $c'$ at or after time $u$.
	Since $a$ performs the exact same movement from time $u$ onwards as if it would have initially started in $c''$ in state $q$, it follows that agent $a$ starting in $c''$ in state $q$ never visits $c'$, which is a contradiction to our assumptions.
	Thus, this last remaining case cannot occur, which completes the proof.
\end{proof}

\begin{lemma}
	Let $(t, u)$ be a travel meeting pair.
	Consider the tuple $Q_t := (q_t^1, q_t^2, q_t^3, c_t^1 \ominus c_t^2, c_t^1 \ominus c_t^3, c_t^2 \ominus c_t^3, \anext{t}, M_t)$, where $\anext{t}$ again denotes the agent that is scheduled at time $t$.
	There is a point in time $T$ such that the following holds:
	If $t \geq T$, then $Q_t$ uniquely determines the tuple $Q_u = (q_u^1, q_u^2, q_u^3, c_u^1 \ominus c_u^2, c_u^1 \ominus c_u^3, c_u^2 \ominus c_u^3, \anext{u}, M_u)$.

	\label{lemma: modulo}
\end{lemma}
\begin{proof}
	Let $T$ be sufficiently large so that $T \geq T_{3N+1}$ holds and Lemma \ref{lemma: once} and Lemma \ref{lemma: vector} apply.
	Let $(t, u)$ be a travel meeting pair with $t \geq T$.
	We start by observing that the subschedule of $\anext{t}$ \emph{starts} at time $t$.
	The reason for this is that if $\anext{t} \in M_t$, then $\anext{t}$ cannot have been scheduled at time $t-1$ as otherwise its subschedule would not continue at time $t$ due to the specification of our schedule, and if $\anext{t} \notin M_t$, then $(t, u)$ would not be a travel meeting pair.	

	Combining the argumentation from the proof of Lemma \ref{lemma: traveltravel} (about the uniquely determined subschedules) with Lemma \ref{lemma: type2} and Lemma \ref{lemma: once}, we see that $Q_t$ uniquely determines which agent is the traveling agent for $(t, u)$, which in turn uniquely determines $M_u$.
	By Lemma \ref{lemma: once}, the last agent that is scheduled before time $u$ is exactly the traveling agent, which uniquely determines $\anext{u}$.	

	Moreover, the information which agent is the traveling agent together with the information which agent is scheduled at time $t$ (i.e., $\anext{t}$) uniquely determines which agents are scheduled between time $t$ and time $u$ (and all of them are scheduled only once, possibly for multiple subsequent time steps, if they are scheduled at all).
	Since no two agents meet between time $t$ and time $u$, it follows that the states of the source agent and the destination agent at time $u$ are uniquely determined by the states of the three agents at time $t$ (and the information which two agents are contained in $M_t$).
	Here we use that the subschedules of those agents (if they are scheduled at all) are of type 2, according to Lemma \ref{lemma: once}.	

	Similarly, the exact vectors (possibly the vector $(0,0)$) by which the source and the destination agent move are uniquely determined by $Q_t$.
	By Lemma \ref{lemma: once}, the subschedule of the traveling agent is of type 1 and ends in the cell that is occupied by the agent not contained in $M_t$.
	Since addition and our modulo operation behave nicely (more specifically, because $\left ((w,z) + (w',z') \right) \pmod{(x,y)} = \left ((w,z) \pmod{(x,y)} + (w',z') \right) \pmod{(x,y)}$ for all integers $w,z, w', z'$), we get that $c_u^1 \ominus c_u^2$, $c_u^1 \ominus c_u^3$, and $c_u^2 \ominus c_u^3$ are uniquely determined by the vectors by which source and destination agent move (combined with the information contained in $Q_t$), and thus also by $Q_t$.

	It remains to show that the state of the traveling agent at time $u$ is uniquely determined by $Q_t$.
	Denote the traveling agent by $a_i$.
	Since $t \geq T_{2N+1}$, the distance between $a_i$ and the destination agent at time $t$ is at least $2N+1$.
	When the subschedule of $a_i$ starts at some time $t \leq t' \leq u$, the destination agent may have moved from its location at time $t$, but since the subschedule of the destination agent is of type 2 (as observed above), it has moved a distance of at most $N$, by Lemma \ref{lemma: type2}.
	Hence, at time $t'$ the distance between $a_i$ and the destination agent is at least $N+1$.	

	Now, let $c, c''$ be cells with a distance of at least $N+1$ to the location $c'$ of the destination agent at time $t'$ and assume that $c' \ominus c = c' \ominus c''$.
	Then, according to the definition of our modulo operation, $c'' - c$ is a (possibly negative) integer multiple of $(x,y)$, and thus also of the travel vector of $a_i$, by Lemma \ref{lemma: vector}.
	Thus, by Lemma \ref{lemma: tuple}, it follows from the above that $q_u^i$ is uniquely determined by $Q_t$.
	Note that although $a_i$ may not be alone in its cell at the time its subschedule starts, we can still apply Lemma \ref{lemma: tuple} since after the first step of $a_i$ it is alone in its cell while all other requirements for Lemma \ref{lemma: tuple} are still satisfied.
	This completes the proof.
\end{proof}

\section{Three Semi-Synchronous Agents Do Not Suffice}

We now conclude our lower bound proof with Theorem \ref{thm:threelower}.
Roughly speaking, Lemma \ref{lemma: modulo} certifies that the behavior of the agents between any two subsequent occurrences of the same fixed information tuple $Q_t$ is reasonably similar.
Since there are only finitely many different $Q_t$ that actually occur, it follows that the behavior of the agents loops, in a very informal sense.
From this, we can derive a contradiction to the assumption that all cells are explored.

\begin{repeattheorem}{thm:threelower}
	Three semi-synchronous agents controlled by a finite automaton are not sufficient to explore the infinite grid.	
\end{repeattheorem}
\begin{proof}
	Suppose for a contradiction that three agents suffice to explore the grid.
	From the definition of a travel meeting pair and Lemma \ref{lemma: meeting}, it follows that there are points in time $t_1 < u_1 \leq t_2 < u_2 \leq t_3 < \dots$ such that $(t_j, u_j)$ is a travel meeting pair for any $j \geq 1$ and for every travel meeting pair $(t', u')$ there is a $j \geq 1$ with $t' = t_j$ and $u' = u_j$.

	Recall the definition of $Q_t$ in Lemma \ref{lemma: modulo}.
	Let $T$ be sufficiently large so that $T \geq T_1$ holds (where $T_1$ is just $T_D$ for $D = 1$) and Lemmas \ref{lemma: once}, \ref{lemma: betweentravels}, \ref{lemma: traveltravel}, \ref{lemma: vector} and \ref{lemma: modulo} apply, and let $k$ be an index such that $t_k \geq T$ and there is a $h > k$ with $h - k$ even and $Q_{t_k} = Q_{t_h}$.
	Such a $k$ must exist since there is only a finite number of tuples of the general form $Q_t$ (after time $T$) and the number of travel meeting pairs is infinite, by Lemma \ref{lemma: meeting}.
	Note that the finiteness of the number of tuples, in particular the finiteness of the (combinations of the) relative locations of the agents modulo $(x,y)$, relies on the fact that the possible travel vectors after time $T$ are restricted by Lemma \ref{lemma: vector}, together with the fact that in the time span given by a travel meeting pair source and destination agent are scheduled for at most $N$ steps, by Lemma \ref{lemma: type2} and Lemma \ref{lemma: once}.

	Consider the sequence $((t_k, u_k), (t_{k+1}, u_{k+1}), \dots, (t_h, u_h))$ of travel meeting pairs, where $h$ is the smallest index such that $h > k$ holds, $h - k$ is even, and $Q_{t_k} = Q_{t_h}$.
	We examine the cells that are explored by the source agent for $(t_k, u_k)$ between time $t_k$ and $t_{k + 1}$ and by the destination agent for $(t_{k+1}, u_{k+1})$ (which is the same as the aforementioned source agent) between time $t_{k + 1}$ and $t_{k + 2}$.
	Then we iterate this examination, in each iteration increasing the indices by $2$, and stop at time $t_h$.
	We say that the cells explored in the described way are explored during \emph{even} explorations.	

	In the first iteration, we obtain the following picture, where we denote the source agent for $(t_k, u_k)$ (i.e., the destination agent for $(t_{k+1}, u_{k+1})$) by $a$:
	The exact vector by which $a$ moves between time $t_k$ and $u_k$ is uniquely determined by $Q_{t_k}$, as observed in the proof of Lemma \ref{lemma: modulo}.
	The exact vector by which $a$ moves between time $u_k$ and $t_{k+1}$ is uniquely determined by $Q_{u_k}$, by Lemma \ref{lemma: traveltravel}.
	Similarly, the exact vectors by which $a$ moves between time $t_{k+1}$ and $u_{k+1}$ and between time $u_{k+1}$ and $t_{k+2}$ are uniquely determined by $Q_{t_{k+1}}$ and $Q_{u_{k+1}}$, respectively.

	Moreover, by combining Lemma \ref{lemma: traveltravel} and Lemma \ref{lemma: modulo}, we see that $Q_{u_k}$, $Q_{t_{k+1}}$, $Q_{u_{k+1}}$, and $Q_{t_{k+2}}$ are all uniquely determined by $Q_{t_k}$.
	Thus, the exact vector by which $a$ moves between time $t_k$ and time $t_{k+2}$ is uniquely determined by $Q_{t_k}$.
	Furthermore, by Lemma \ref{lemma: type2}, Lemma \ref{lemma: once}, and Lemma \ref{lemma: betweentravels}, the number of cells $a$ visits between time $t_k$ and time $t_{k+2}$ is bounded by a constant.
	Note that each $Q_{t_j}$ also uniquely determines which agent is the traveling agent (and hence which agent is the source/destination agent) for $(t_j, u_j)$, as observed in the proof of Lemma \ref{lemma: modulo}.

	For the second, third, $\dots$, iteration we obtain an analogous picture.
	Hence, the tuples $Q_{t_{k + 2}}, Q_{t_{k + 4}}, \dots$ are all uniquely determined by $Q_{t_k}$, and the locations of the respective source agents at times $t_{k + 2}, t_{k + 4}, \dots$ are all uniquely determined by $Q_{t_k}$ and the location of the source agent for $(t_k, u_k)$ at time $t_k$.


	We obtain the following bigger picture:
	The location of the source agent for $(t_k, u_k)$ at time $t_k$ together with $Q_{t_k}$ uniquely determines both $Q_{t_h}$ and the location of the source agent for $(t_h, u_h)$ at time $t_h$, which, in turn, uniquely determine $Q_{t_{h+(h-k)}}$ and the location of the source agent for $(t_{h+(h-k)}, u_{h+(h-k)})$ at time $t_{h+(h-k)}$, and so on.
	Hence, there is a vector $(w,z)$ such that the locations of the respective source agents at times $t_k, t_h, t_{h+(h-k)}, t_{h+2(h-k)}, \dots$ are $c, c + (w,z), c + 2(w,z), \dots$, where $c$ denotes the cell occupied by the respective source agent at time $t_k$.
	Moreover, since the number of cells explored during an even exploration between time $t_k$ and $t_h$ (and similarly between time $t_{h+j(h-k)}$ and $t_{h+(j+1)(h-k)}$ for each $j \geq 0$) is bounded by a constant (which follows from a similar observation above), we obtain that there is a constant $L$ such that each cell explored during an even exploration has a distance of at most $L$ to some cell of the form $c + j' \cdot (w,z)$, where $j'$ is some non-negative integer.

	Moreover, by Lemmas \ref{lemma: type2}, \ref{lemma: once}, \ref{lemma: betweentravels}, \ref{lemma: vector}, and the definition of even explorations, we know that each explored cell is close to the travel of a traveling agent, i.e., there is a constant $L'$ such that each cell explored at or after time $t_k$ has a distance of at most $L'$ to some cell of the form $c' + j'' \cdot (x,y)$, where $j''$ is some integer and $c'$ a cell explored during an even exploration. 
	Combining our observations and adding the fact that only a constant number of cells are explored up to time $t_k$, it follows that there is a constant $L''$ such that each cell explored by the agents has a distance of at most $L''$ to some cell of the form $c + j' \cdot (w,z) + j'' \cdot (x,y)$, where $j', j''$ are integers and $j'$ is non-negative.
	Hence, we can draw a line in the grid such that all explored cells are to one side of the line, yielding a contradiction to the assumption that three agents suffice to explore the grid.	
\end{proof}

\section{Conclusion}
In this paper, we considered the collaborative grid exploration problem with agents controlled by asynchronous finite automata.
While this paper shows a tight bound for the minimum number of agents required to explore the grid, it remains an intriguing open question to generalize these results, especially finding non-trivial lower bounds, to more complex structures. 
From previous work, we know that there are graphs where no finite set of agents suffices~\cite{Rollik1979}.
However, several natural graph classes fall between grids and general graphs, such as planar graphs and penta/hexagrids.
It would be interesting to obtain a better understanding of the dependence of the required number of agents on the choice of the underlying (intermediate) graph structure.

\bibliography{references}

\end{document}